\documentclass[11pt]{article}

\usepackage{amsmath}
\usepackage{amsthm}
\usepackage{amssymb}
\usepackage{amsfonts}
\usepackage{arcs}
\usepackage{float}
\usepackage{graphicx}
\usepackage[small]{caption}
\usepackage{color}
\usepackage{enumitem}
\usepackage{marvosym}
\usepackage{cases}
\usepackage{url}
\usepackage{hyperref}
\usepackage{pifont}
\usepackage{graphicx}

\newtheorem{theorem}{Theorem}
\newtheorem{lemma}{Lemma}
\newtheorem{corollary}{Corollary}

\makeatletter
\DeclareFontFamily{U}{tipa}{}
\DeclareFontShape{U}{tipa}{m}{n}{<->tipa10}{}
\newcommand{\arc@char}{{\usefont{U}{tipa}{m}{n}\symbol{62}}}%

\newcommand{\arc}[1]{\mathpalette\arc@arc{#1}}

\newcommand{\arc@arc}[2]{%
	\sbox0{$\m@th#1#2$}%
	\vbox{
		\hbox{\resizebox{\wd0}{\height}{\arc@char}}
		\nointerlineskip
		\box0
	}%
}
\makeatother

\newcommand{\NN}{\mathbb{N}} 
\newcommand{\ZZ}{\mathbb{Z}} 
\newcommand{\RR}{\mathbb{R}} 

\newcommand{\eps}{\varepsilon}

\newcommand{\etal}{{et~al.}}
\newcommand{\ie}{{i.e.}}
\newcommand{\eg}{{e.g.}}

\def\T{\mathcal T}

\def\diam{\texttt{diam}}

\newcommand{\later}[1]{{}}
\newcommand{\old}[1]{{}}
\long\def\ignore#1{}

\begin{document}
	
\title{ {\sc Sparse Hop Spanners for Unit Disk Graphs}\footnote{A
  preliminary version of this paper appears
  in the Proceedings of the 31st International Symposium on Algorithms and Computation (ISAAC 2020),
  LIPIcs, vol. 181, Schloss Dagstuhl - Leibniz-Zentrum f{\"{u}}r
  Informatik, 2020, pp. 57:1--57:17.}}
	
\author{%
  Adrian Dumitrescu\footnote{\url{adriandumitrescu.org} Email: \texttt{ad.dumitrescu@gmail.com}.}
\and
  Anirban Ghosh\footnote{School of Computing, University of North Florida, Jacksonville, FL, USA.
	Email: \texttt{anirban.ghosh@unf.edu}.
	Research on this paper was partially supported by the University of North Florida
        Academic Technology Grant and by the NSF award CCF-1947887.}
  \and
  Csaba D. T\'oth\footnote{Department of Mathematics, California State University Northridge,
	Los Angeles, CA; and Department of Computer Science, Tufts University, Medford, MA, USA.
	Email:~\texttt{cdtoth@acm.org}.
        Research on this paper was partially supported by the NSF award DMS-1800734.}
	}
	
\maketitle

\begin{abstract}
A unit disk graph $G$ on a given set $P$ of points in the plane is a
geometric graph where an edge exists between two points $p,q \in P$
if and only if $|pq| \leq 1$.  A spanning subgraph $G'$ of $G$ is a
\emph{$k$-hop spanner} if and only if for every edge $pq\in G$, there is a path
between $p,q$ in $G'$ with at most $k$ edges.
We obtain the following results for unit disk graphs in the plane.

\begin{enumerate} [label=(\roman*)] \itemsep 0pt
\item Every $n$-vertex unit disk graph has a $5$-hop spanner with
  at most $5.5n$ edges. We analyze the family of spanners constructed by
  Biniaz (2020) and improve the upper bound on the number of edges
  from $9n$ to $5.5n$.

\item Using a new construction, we show that every $n$-vertex unit disk graph has
  a $3$-hop spanner with at most $11n$ edges.

\item Every $n$-vertex unit disk graph has a $2$-hop spanner with $O(n\log n)$ edges.
  This is the first nontrivial construction of $2$-hop spanners.

\item For every sufficiently large positive integer $n$, there exists a set $P$ of
  $n$ points on a circle, such that every plane hop spanner on $P$ has hop stretch factor
  at least $4$. Previously, no lower bound greater than $2$ was known.

\item For every finite point set on a circle, there exists a plane (\ie, crossing-free) $4$-hop spanner.
  As such, this provides a tight bound for points on a circle.

\item The maximum degree of $k$-hop spanners cannot be bounded from above by
      a function of~$k$ for any positive integer $k$.
\end{enumerate}

\end{abstract}

\section{Introduction}  \label{sec:intro}

A $k$-spanner (or \emph{$k$-hop spanner}) of a connected graph $G=(V,E)$ is a
subgraph $G'=(V,E')$, where $E' \subseteq E$, with the additional property that the distance
between any two vertices in $G'$ is at most $k$ times the distance
in $G$~\cite{kortsarz1994generating,peleg1989graph}, where the \emph{distance} between two
vertices is the minimum number of edges on a path between them.
The graph $G$ itself is a $1$-hop spanner. The minimum $k$ for which a subgraph
$G'$ is a $k$-spanner of $G$ is referred to as the \emph{hop stretch factor}
(or \emph{hop number}) of $G'$. An alternative characterization of $k$-spanners
is given in the following lemma.

\begin{lemma}[Peleg and Sch\"affer~\cite{peleg1989graph}] \label{lem:equiv}
The subgraph $G'= (V,E')$ is a $k$-spanner of the graph $G=(V,E)$ if and only if
the distance between $u$ and $v$ in $G'$ is at most $k$ for every edge $uv \in E$.
\end{lemma}

If the subgraph $G'$ has only $O(|V|)$ edges, then $G'$ is called
a \emph{sparse} spanner. In this paper we are concerned with constructing
sparse $k$-spanners (with small $k$) for unit disk graphs in the plane.
Given a set $P$ of $n$ points $p_1,\ldots,p_n$ in the plane,
the unit disk graph (UDG) is a geometric graph $G=G(P)$ on the vertex set
$P$ whose edges connect points that are at most unit distance apart.
A \emph{spanner of a point set} $P$ is a spanner of its UDG.

Recognizing UDGs was shown to be NP-Hard by Breu and Kirkpatrick~\cite{breu1998unit}.
Unit disk graphs are commonly used to model network topology in ad hoc wireless
networks and sensor networks.
They are also used in multi-robot systems for practical purposes such as
planning, routing, power assignment, search-and-rescue, information collection,
and patrolling; refer to~\cite{alzoubi2003geometric, dutta2019multi,
  kanj2008geometric,DBLP:journals/ijcga/LiW04,nieberg2008approximation}
for some applications of UDGs.
For packet routing and other applications,
a bounded-degree plane geometric spanner of the wireless network
is often desired but not always feasible~\cite{bose2005constructing}.
Since a UDG on $n$ points can have a quadratic number of edges, a common desideratum
is finding sparse subgraphs that approximate the respective UDG with respect to
various criteria.
\emph{Plane} spanners, in which no two edges cross, are desirable
for applications where edge crossings may cause interference.

Obviously, for every $k \geq 1$, every graph $G=(V,E)$ on $n$ vertices has a
$k$-spanner with $|E|=O(n^2)$ edges. If $G$ is the complete graph,
a star rooted at any vertex is a $2$-hop spanner with $n-1$ edges.
However, the $O(n^2)$ bound on the size of a $2$-hop spanner cannot be
improved; a classic example~\cite{kortsarz1994generating} is that of a
complete bipartite graph with $n/2$ vertices on each side.
In general, if $G$ has girth $k+2$ or higher, then its only $k$-spanner is $G$ itself.
According to Erd\H{o}s' girth conjecture~\cite{Erdos64extremalproblems}, the maximum size of a graph
with $n$ vertices and girth $k+2$ is $\Theta(n^{1+1/\lceil k/2\rceil})$ for $k\geq 2$.
The conjecture has been confirmed for some small values of $k$, but remains open for $k>9$.
For any graph $G$ with $n$ vertices, a $k$-spanner with $O(n^{1+1/\lceil k/2\rceil})$ edges
can be constructed in linear time~\cite{baswana2007simple,Baswana2016}.
We show that for unit disk graphs, we can do much better in terms
of the number of edges for every $k \geq 2$.

Spanners in general and unit disk graph spanners in particular are used
to reduce the size of a network and the amount of routing information.
They are also used for maintaining network connectivity, improving throughput,
and optimizing network lifetime~\cite{biniaz2020plane,gao2005geometric,kanj2008geometric,
  li2003algorithmic,rajaraman2002topology}.

Spanners for UDGs with hop stretch factors bounded by a constant
were introduced by Catusse, Chepoi, and Vax{\`{e}}s in~\cite{catusse2010planar}.
They constructed (i)~$5$-hop spanners with at most $10n$ edges for $n$-vertex UDGs;
and (ii)~plane $449$-hop spanners with less than $3n$ edges.
Recently, Biniaz~\cite{biniaz2020plane} improved both these results,
and showed that for every $n$-vertex unit disk graph, there exists (i)~a $5$-hop
spanner with at most $9n$ edges, and (ii)~a plane $341$-hop spanner.
The algorithms presented in~\cite{biniaz2020plane,catusse2010planar}
run in time that is polynomial in $n$. A summary of these results and our new results
is included in Table~\ref{tab:summary}.

\begin{table}[ht]
\centering
\begin{tabular}{|c|c|c| c|}
\hline
Reference & $k$ & $|E'|$ & Guaranteed to be plane?\\
\hline
\hline
{ Catusse, Chepoi, and Vax{\`{e}}s (2010)}~\cite{catusse2010planar} &
$5$ & $\leq 10n$ & {\large \ding{55}}  \\
\hline
{ Catusse, Chepoi, and Vax{\`{e}}s (2010)}~\cite{catusse2010planar} &
$449$ & $\leq 3n$ & {\large \ding{52}}  \\
\hline
{ Biniaz (2020)}~\cite{biniaz2020plane}	& $5$ & $\leq 9n$ & {\large \ding{55}}  \\	
\hline
{ Biniaz (2020)}~\cite{biniaz2020plane}	& $341$ & $\leq 3n$ & {\large \ding{52}}  \\	
\hline
{\emph{This paper}}  & $5$ & $\leq 5.5n$ & {\large \ding{55}}  \\		
\hline
{\emph{This paper}}  & $3$ & $\leq 11n$ & {\large \ding{55}}  \\		
\hline
{\emph{This paper}}  & $2$ & $O(n\log n)$ & {\large \ding{55}}  \\		
\hline
\end{tabular}
\caption{A summary of results on constructions of hop spanners for unit disk graphs in the plane.}
\label{tab:summary}
\end{table}

\paragraph{Our results.}
The following are shown for unit disk graphs.
\begin{enumerate} [label=(\roman*)] \itemsep 0pt

\item Every $n$-vertex unit disk graph has a $5$-hop spanner with at most $5.5n$ edges
  (Theorem~\ref{thm:biniaz} in Section~\ref{sec:non-plane}). We carefully analyze the
  construction proposed by Biniaz~\cite{biniaz2020plane}
    and improve the upper bound on the number of edges from the $9n$ to $5.5n$.

\item Using a new construction, we show that every $n$-vertex
    unit disk graph has a $3$-hop spanner with at most $11n$ edges
    (Theorem~\ref{thm:3hop} in Section~\ref{sec:non-plane}).
    Previously, no $3$-hop spanner construction algorithm was known.

\item Every $n$-vertex unit disk graph has a $2$-hop spanner with $O(n\log n)$ edges.
  This is the first construction with a subquadratic number of edges
  (Theorem~\ref{thm:2hop+} in Section~\ref{sec:2hop}) and our main result.

\item For every $n\geq 8$, there exists an $n$-element point set $P$ such
      that every plane hop spanner on $P$ has hop stretch factor at least $3$.
      If $n$ is sufficiently large, the lower bound can be raised to~$4$
      (Theorems~\ref{thm:8points} and~\ref{thm:many-points} in Section~\ref{sec:lb}).
      A trivial lower bound of $2$ can be easily obtained by placing
      four points at the four corners of a square of side-length $1/2$.

\item For every finite point set $P$ on a circle $C$, there exists a plane $4$-hop spanner
 (Theorem~\ref{thm:circle} in Section~\ref{sec:lb}). The lower bound of $4$
  holds for some point-set on a circle.

\item For every pair of integers $k \geq 2$ and $\Delta \geq 2$, there exists a set $P$ of
  $n=O(\Delta^k)$ points in the plane such that the unit disk graph
  $G=(P,E)$ on $P$ has no $k$-spanner
  whose maximum degree is at most~$\Delta$ (Theorem~\ref{thm:K_n} in Section~\ref{sec:constant}).
  An extension to dense graphs is given by Theorem~\ref{thm:dense} in Section~\ref{sec:constant}.
  In contrast, Kanj and Perkovi\'c~\cite{kanj2008geometric} showed that UDGs admit bounded-degree
  \emph{geometric} spanners.
\end{enumerate}

\paragraph{Related work.}
Peleg and Sch\"affer~\cite{peleg1989graph} have shown that for a given graph $G$
(not necessarily a UDG)
and a positive integer $m$, it is NP-complete to decide whether there exists a $2$-spanner of $G$
with at most $m$ edges. They also showed that for every graph on $n$ vertices, a $(4k+1)$-spanner
with $O(n^{1+1/k})$ edges can be constructed in polynomial time. In particular,
every graph on $n$ vertices has a $O(\log{n})$-spanner with $O(n)$ edges.
Their result was improved by Alth\"ofer~\etal~\cite{althofer1993sparse}, who showed
that a $(2k-1)$-spanner with $O(n^{1+1/k})$ edges can be constructed in polynomial time;
the run-time was later improved to linear~\cite{baswana2007simple,bose2013plane}.
Kortsarz and Peleg obtained approximation algorithms for the problem of finding,
in a given graph, a $2$-spanner of minimum size~\cite{kortsarz1994generating} or minimum maximum
degree~\cite{kortsarz1998generating}.

In the geometric setting, where the vertices are embedded in a metric space,
spanners have been studied in~\cite{aronov2008sparse,calinescu2006bounded,
  das1989triangulations,dobkin1990delaunay,levcopoulos1992there,DBLP:journals/ijcga/LiW04}
and many other papers. In particular, plane geometric spanners were studied
in~\cite{bose2005constructing,bose2013plane,dumitrescu2016lattice,dumitrescu2016lower}.
The reader is also referred to the
surveys~\cite{bose2013plane,DBLP:books/el/00/Eppstein00,mitchell2000geometric}
and the monograph~\cite{narasimhan2007geometric} dedicated to this subject.

\paragraph{Notation and terminology.}
For two points $p,q \in \RR^2$, we denote the Euclidean distance by $d(p,q)$
or sometimes by $|pq|$. The distance between two sets, $A,B \subset \RR^2$, is defined by
$d(A,B) =\inf \{d(a,b) : a \in A, b \in B\}$.
The diameter of a set $A$, denoted $\diam(A)$, is defined by
$\diam(A) =\sup \{d(a,b) : a,b \in A\}$.
For a set $A$, its boundary and interior
are denoted by $\partial A$ and $\text{int}(A)$, respectively.

A geometric graph $G=(P,E)$ is a \emph{geometric $t$-spanner}, for some $t \geq 1$,
if for every pair of vertices $u,v\in P$, the minimum Euclidean length of a path $\pi_G(u,v)$
between $u$ and $v$ in $G$ is at most $t$ times $|uv|$, \ie, $\forall
u,v \in V, |\pi_G(u,v)| \leq t |uv|$.
When there is no necessity to specify $t$, we simply use the term \emph{geometric spanner}.

Given a graph $G=(V,E)$ and a vertex $u\in V$, the \emph{neighborhood} $N(u)$
is the set of vertices adjacent to $u$. For brevity, a \emph{hop spanner for} a point set
$P\subset \RR^2$ is a hop spanner for the UDG on $P$.
Assume we are given a subgraph $G'=(P,E')$ of the UDG for a point set $P$.
For $p,q \in P$, let $\rho(p,q)$ denote a shortest path in $G'$,
\ie, a path containing the fewest edges; and $h(p,q)$ denote the
corresponding hop distance (number of edges).

A geometric graph is \emph{plane} if any two distinct edges are either disjoint
or only share a common endpoint. Whenever we discuss plane graphs (plane spanners in particular),
we assume that the points (vertices) are in \emph{general position}, \ie, no three
points are collinear.

A \emph{unit disk} (resp., \emph{circle}) is a disk (resp., circle) of unit radius.
The complete bipartite graph with parts of size $m$ and $n$ is denoted by $K_{m,n}$;
in particular, $K_{1,n}$ is a star on $n+1$ vertices.
We use the shorthand notation $[n]$ for the set $\{1,2,\ldots,n\}$.

\section{Sparse (possibly nonplane) hop spanners}  \label{sec:non-plane}

In this section we construct hop spanners with a linear number of edges
that provide various trade-offs between the two parameters of interest:
number of hops and number of edges.

\subsection{Construction of 5-hop spanners} \label{ssec:5hop}

We start with a short outline of the $5$-hop spanner
constructed by Biniaz~\cite[Theorem~3]{biniaz2020plane}; it is based on a regular hexagonal tiling
of the plane with cells of unit diameter. Hence the UDG of a finite point set
$P\subset \RR^2$
contains every edge between points in the same cell. In every nonempty cell, a star
rooted at an arbitrarily chosen point in the cell is created.
Then, for every pair of cells, exactly one edge of the UDG is chosen,
if such an edge exists. Biniaz showed that the resulting graph is
a $5$-hop spanner with at most $9n$ edges.

We next provide a more detailed description and an improved analysis of the
construction. Consider a regular hexagonal tiling $\T$ in the plane with cells of unit diameter;
refer to Fig.~\ref{fig:hex}\,(left). Let $P$ be a finite set of points in the plane.
We may assume that no point in $P$ lies on a cell
boundary. Every point in $P$ lies in the interior of some cell
of $\T$ (and so the distance between any two points inside a cell is
less than~$1$).
Let $p \in P$ be a point in a cell $\sigma$. Denote by
$H_1,\ldots,H_6$ the six cells adjacent to $\sigma$ in
counterclockwise order; these cells form the \emph{first layer} around
$\sigma$. Let $H_7,\ldots,H_{18}$ be
the twelve cells at distance two from $\sigma$ in counterclockwise order, forming
the \emph{second layer} around $\sigma$, such that $H_7$ is adjacent to only
$H_1$ in the first layer.

\begin{figure}[htbp]
\centering
\includegraphics[scale=0.56]{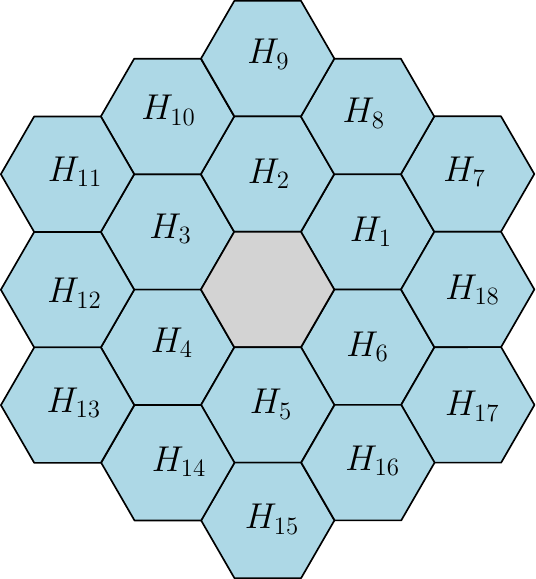}
\hspace{2cm}
\includegraphics[scale=0.56]{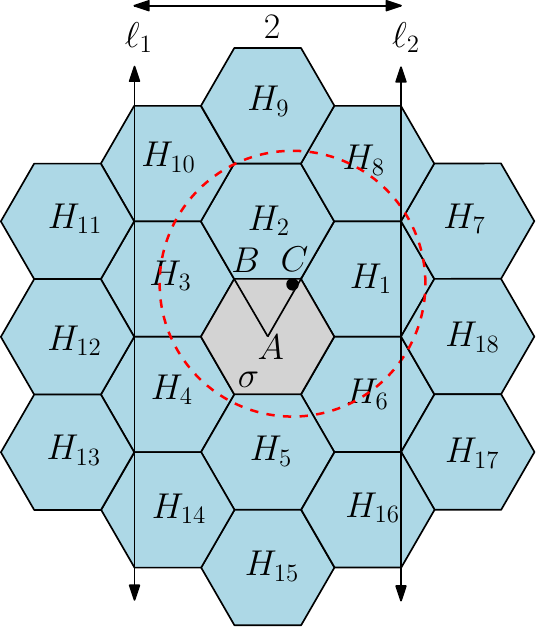}
\caption{Left: A regular hexagonal tiling with cells of unit diameter;
  the figure shows the two layers
  of cells around $\sigma$. Right: The unit disk centered at $p$
  intersects $11$ cells $H_1,\ldots,H_{10},H_{18}$.}
\label{fig:hex}
\end{figure}

For every two distinct cells $\sigma,\tau\in\T$, take an arbitrary edge $pq \in E$,
$p \in \sigma$, $q \in \tau$, if such an edge exists; we call such an edge a \emph{bridge}.
Each cell $\sigma$ can have bridges to at most $18$ other cells, namely those
in the two layers around $\sigma$. A bridge is \emph{short} if it connects points in adjacent cells
and \emph{long} otherwise.

\begin{lemma}\label{lem:leq5}
  Let $p \in P$ be a point that lies in cell $\sigma$.
  The unit disk $D$ centered at $p$ intersects at most five cells from the second layer
  around $\sigma$.
\end{lemma}
\begin{proof}
Let $A$ be the center of $\sigma$ (shaded gray in Fig.~\ref{fig:hex}\,(right)).
Subdivide $\sigma$ into six regular triangles incident to $A$.
By symmetry, we can assume that $p \in \Delta{ABC}$, where $BC=\sigma\cap H_2$.

Note that $d(\Delta{ABC},H_i)>1$ for $i\in \{13,14,15,16,17\}$,
and $D$ is disjoint from the five cells $H_{13}$, $H_{14}$, $H_{15}$, $H_{16}$, and $H_{17}$.
Now, observe that $d(H_{7} \cup H_{18},H_{11} \cup H_{12})=2$.
Hence, $D$ intersects at most one of
$H_{7} \cup H_{18}$ and $H_{11} \cup H_{12}$.
Consequently, $D$ intersects at most $12-5-2=5$ cells
from the second layer around $\sigma$.
\end{proof}

Obviously, any two points in a cell $\sigma$ are at most unit distance apart.
Further, observe that the unit disk $D$ centered at $p$ intersects all six
cells $H_1,\ldots,H_6$.
As such, Lemma~\ref{lem:leq5} immediately yields the following.

\begin{corollary}\label{cor:leq11}
  All neighbors of each point $p \in \sigma$ lies in $\sigma$ and at most $11$ cells around $\sigma$.
\end{corollary}

\begin{theorem} \label{thm:biniaz}
The (possibly nonplane) $5$-hop spanner constructed by Biniaz~\cite[Theorem~3]{biniaz2020plane}
has at most $5.5n$ edges.
\later{Such a spanner can be computed in $O(n \log{n})$ time.
---comment---Algorithm analyses and running times are missing at the moment.}
\end{theorem}
\begin{proof}
Let $P$ be a set of $n$ points and $G=(P,E)$ be the corresponding UDG.
Let $x \geq 1$ be the number of points in a hexagonal cell $\sigma\in\T$.
The construction has $x-1$ inner edges that make a star and at most
$18$ outer edges (bridges) connecting points in $\sigma$ with points in other cells.
We analyze the situation depending on $x$.

If $x=1$, there are no inner edges and at most $11$ outer edges by Corollary~\ref{cor:leq11}.
As such, the degree of the (unique) point in $\sigma$ is at most $11$.

If $x=2$, there is one inner edge and at most $16$ outer edges.
Indeed, by Lemma~\ref{lem:leq5}, each point $p\in P\cap
\sigma$ has neighbors in at most five cells
from the second layer around $\sigma$ (besides points in $P$
in the six cells in the first layer).
Two points in $P\cap \sigma$ can jointly have neighbors in at most $6+5+5=16$ other cells.
As such, the average degree for points in $\sigma$ is at most $(2+16)/2=9$.

If $x \geq 3$, there are $x-1$ inner edges and at most $18$ outer edges.
As such, the average degree for points in $\sigma$ is at most
\[ \frac{2(x-1)+18}{x} = \frac{2x+16}{x} \leq \frac{22}{3}. \]

Summation over all cells implies that the average degree in the resulting
$5$-hop spanner $G'$ is at most $11$, 
thus $G'$ has at most $5.5n$ edges.
\later{
Once the grid is chosen, the points assigned to each cell can be computed
in $O(n)$ time by a bucketing technique.
The remaining steps can be executed in $O(n \log{n})$ time
} 
\end{proof}

\subsection{Construction of 3-hop spanners} \label{ssec:3hop}

Here we show that every point set in the plane has a $3$-hop spanners of linear size.
This brings down the hop-stretch factor of Biniaz's construction from $5$
to $3$ at the expense of increasing in the number of edges (from $5.5n$ to $11n$).

\begin{theorem} \label{thm:3hop}
  Every $n$-vertex unit disk graph has a (possibly nonplane)
  $3$-hop spanner with at most $11n$ edges.
\end{theorem}
\begin{proof}
Let $P$ be a set of $n$ points in the plane, and let
$G=(P,E)$ be the UDG of $P$.
Let $G'$ be the 5-hop spanner described in Section~\ref{ssec:5hop},
based on a hexagonal tiling $\T$ with cells of unit diameter.
We construct a new graph $G''$ that consists of all bridges from $G'$
and, for each nonempty cell $\sigma\in \T$, a spanning star of the points in $\sigma$
defined as follows.

\begin{figure}[htbp]
 \centering
 \includegraphics[scale=0.55]{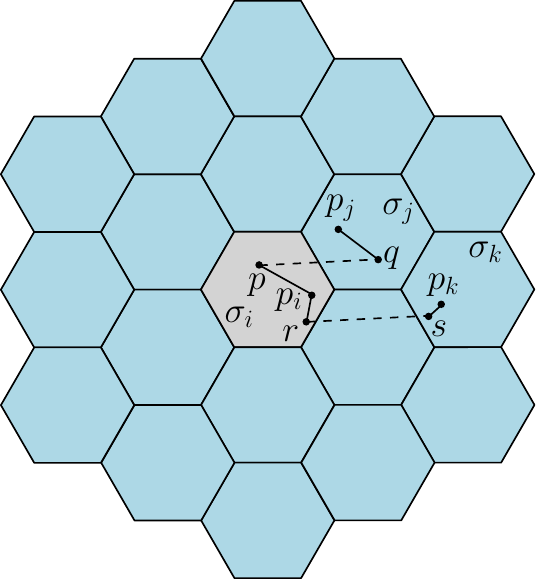}
 \caption{Three points in $P$, $p_i \in \sigma_i$, $p_j \in \sigma_j$, and $p_k \in \sigma_k$
  where $p_i p_j, p_i p_k \in E$. Edge $pq$ is a short bridge connecting $\sigma_i$ and $\sigma_j$
  and edge $rs$ is a long bridge connecting $\sigma_i$ and $\sigma_k$.}
  \label{fig:hex2}
\end{figure}

Let $\sigma\in \T$ be a nonempty cell and let $p_i \in P \cap \sigma$.
For every cell $\tau\in \T$ in the two layers around $\sigma$,
if $d(p_i,\tau)\leq 1$ and $G'$ contains a bridge $pq$,
where $p\in \sigma\setminus \{p_i\}$ and $q\in \tau$,
then we add the edge $p_ip$ to $G''$.
Since $\diam(\sigma)=1$, if $pq$ is a short bridge,
then $p$ is the center of a spanning star on $P \cap \sigma$.
In addition, if no short bridge is incident to any point in $\sigma$,
then we add a spanning star of $P\cap \sigma$
(centered at the endpoint of a long bridge, if any) to $G''$.

It is easy to see that the hop distance between any two points within a cell is at most~$2$.
Indeed, by construction, the points in each nonempty cell are connected by a spanning star.
Consider now a pair of points $p_i \in P \cap \sigma_i$, $p_j
\in P \cap \sigma_j$, $i \neq j$,
where $p_i p_j \in E$. By construction, there is a bridge $pq \in G''$
between the cells $\sigma_i$ and $\sigma_j$. As such, $p_i$ is connected to $p_j$
by a $3$-hop path $p_i,p,q,p_j$. Refer to Fig.~\ref{fig:hex2} for an illustration.

We can bound the average degree of the points in $\sigma$ as follows.
Let $x$ be the number of points in $\sigma$.
By Corollary~\ref{cor:leq11}, the neighbors of each point $p_i\in \sigma$
lie in $\sigma$ and at most 11 cells around $\sigma$.
If $p_i$ is not incident to any bridge, we add at most $11$ edges between $p_i$
and other points in $\sigma$; these edges increase the sum of degrees in $\sigma$ by $2\cdot 11=22$.
Otherwise assume that $p_i$ is incident to $b_i$ bridges, for some $1\leq b_i\leq 11$.
Then we add edges from $p_i$ to at most $11-b_i$ other points in $\sigma$.
The $b_i$ bridges each have only one endpoint in $\sigma$. Overall, these edges
contribute $2(11-b_i)+b_i=22-b_i< 22$ to the sum of degrees in $\sigma$.

If no short bridge has an endpoint in $\sigma$, then by Lemma~\ref{lem:leq5} we add
at most 5 edges between each point $p_i\in \sigma$ and endpoints of long bridges;
these edges increase the sum of degrees in $\sigma$ by $2\cdot 5=10$.
However, we also add a spanning star that contributes $2(x-1)$ to the same sum.
Overall, the sum of degrees in $\sigma$ is bounded from above by
\[
  \begin{cases}
   2\cdot 11x  = 22x,    &  \text{if some short bridge has an endpoint in $\sigma$}\\
    2(x-1)+10x < 12x, & \text{otherwise.}
  \end{cases}
\]
Thus, the average vertex degree is at most $22$ in all $\sigma\in \T$.
Consequently, the $3$-hop spanner $G''$ has at most $11n$ edges.
\later{
The running time remains $O(n)$; similarly to the proof of Theorem~\ref{thm:3hop}.
} 
\end{proof}

\paragraph{Remark.} It is natural to ponder whether the UDG on any $n$ points in the plane
has a subgraph with $O(n)$ edges that is  a $k$-hop spanner (for small $k$) and also
a geometric spanner of $G$.
Such subgraphs of UDGs can find practical uses in the real-world. Interestingly, the answer is yes.
It is shown by Kanj and Perkovi\'c~\cite{kanj2008geometric} that the UDG of a point set $P$
has a subgraph $G_1=(P,E_1)$ with $O(n)$ edges that is a
geometric $t$-spanner for some constant $t$.
Let $G_2=(P,E_2)$ be the $3$-hop spanner generated by the
construction in Theorem~\ref{thm:3hop}.
Clearly, the graph $G':=(P,E_1 \cup E_2)$ is a subgraph of the
UDG, it has $O(n)$ edges, and it is
both a $3$-hop spanner for the UDG of $P$ and a geometric
$t$-spanner for $P$ with a constant $t$.

\section{Construction of 2-hop spanners}   \label{sec:2hop}

In this section, we construct a 2-hop spanner with $O(n\log n)$ edges
for a set $P$ of $n$ points in the plane.
We begin with a construction in a bipartite setting (cf.~Lemma~\ref{lem:bipartite+}),
and then extend it to the general setup.

We briefly review the concept of $\eps$-nets~\cite{mustafa2017epsilon}, which is crucial
for our construction.
Let $(P,\mathcal{R})$ be a set system (a.k.a. \emph{range
  space}), where $P$ is a finite set in an ambient space
and $\mathcal{R}$ is a collection of subsets of that space (called \emph{ranges}).
For $\eps>0$, an \emph{$\eps$-net} for $(P,\mathcal{R})$ is a set $N\subset P$
such that for every $R\in \mathcal{R}$, $|P\cap R|\geq
\eps\cdot |P|$ implies $N\cap R\neq \emptyset$.
When the ambient space is $\mathbb{R}^d$ for some $d\in \mathbb{N}$, and $\mathcal{R}$ is a collection
of semi-algebraic sets, there exists an $\eps$-net of size $O(\frac{d}{\eps}\log \frac{d}{\eps})$,
and this bound is best possible in many cases~\cite{pach2013tight}.
However, for some geometric set systems, $\eps$-nets of size $O(\frac{1}{\eps})$ are possible.
For example, if $P$ is a set of points in the plane and $\mathcal{R}$ consists of halfplanes,
then there exists an $\eps$-net of size $O(\frac{1}{\eps})$~\cite{pach1990some}.
We adapt this results to unit disks in a somewhat stronger form (cf.~Lemma~\ref{lem:epsilon}).

\begin{figure}[htbp]
 \centering
 \includegraphics[width=0.8\textwidth]{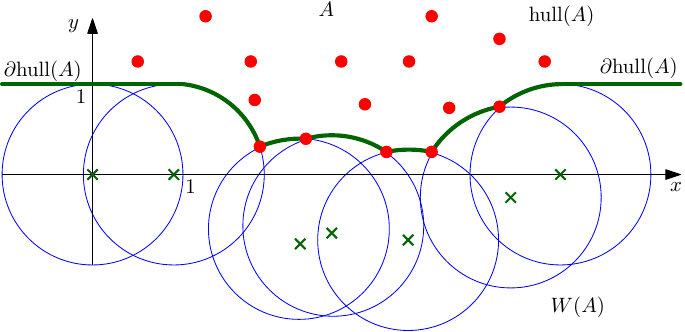}
 \caption{A set $A$ of 16 points above the $x$-axis, $W(A)$, and $\text{hull}(A)$.
   The boundary $\partial \text{hull}(A)$ is an $x$-monotone curve,
   which consists of horizontal segments and arcs of unit circles centered
   on or below the $x$-axis (the centers are marked with crosses).}
  \label{fig:alpha}
\end{figure}

\paragraph{Alpha-shapes.}
As a generalization of convex hulls of a set of points,
Edelsbrunner, Kirkpatrick, and Seidel~\cite{edelsbrunner1983shape} introduced
$\alpha$-shapes, using balls of radius $1/\alpha$ instead of halfplanes.
We introduce a similar concept, in the bipartite
setting, as follows; see Fig.~\ref{fig:alpha} for an illustration.
We consider the set system $(A,\mathcal{D})$, where $A$ is a finite set of points
in the plane above the $x$-axis and $\mathcal{D}$ is the set of all unit disks centered
on or below the $x$-axis.
Let $W(A)$ be the union of all unit disks $D\in \mathcal{D}$ such that
$A\cap \text{int}(D)=\emptyset$;
and let $\text{hull}(A)=\mathbb{R}^2\setminus \text{int}(W(A))$.

The following easy observation shows that disks in $\mathcal{D}$,
restricted to the upper halfplane $\{(x,y)\in \mathbb{R}^2: y>0\}$,
behave similarly to halfplanes in $\mathbb{R}^2$.

\begin{lemma} \label{lem:axiom}
For any two points $p_1,p_2\in \mathbb{R}^2$ above the $x$-axis, there
is at most one unit circle centered at a point on or below the
$x$-axis that is incident to both $p_1$ and $p_2$. Consequently, for
any two unit disks $D_1,D_2\in \mathcal{D}$, at most one point in
$\partial D_1\cap \partial D_2$ lies above the $x$-axis.
\end{lemma}
\begin{proof}
Suppose that two unit circles, $c_1$ and $c_2$, are incident to both
$p_1$ and $p_2$. Then the centers of $c_1$ and $c_2$ are on the
orthogonal bisector of segment $p_1p_2$, on opposite sides of the line
through $p_1p_2$. Hence one of the circle centers is above the
$x$-axis. Therefore at most one of the circles is centered at a point
on or below the $x$-axis.
\end{proof}

We continue with a few basic properties of the boundary of $\text{hull}(A)$,
which exhibit the same behavior as convex hulls with respect to lines in the plane.

\begin{lemma} \label{lem:hull}
The set system $(A,\mathcal{D})$ defined above has the following properties:
\begin{enumerate}\itemsep0pt
\item $\partial \text{hull}(A)$ lies above the $x$-axis;
\item every vertical line intersects $\partial \text{hull}(A)$ in one
  point, thus $\partial \text{hull}(A)$ is an $x$-monotone curve;
\item for every unit disk $D\in \mathcal{D}$, the intersection $D\cap (\partial \text{hull}(A))$
  is connected (possibly empty);
\item for every unit disk $D\in \mathcal{D}$, if $A\cap D\neq \emptyset$,
  then $A\cap D$ contains a point in $\partial \text{hull}(A)$.
\end{enumerate}
\end{lemma}
\begin{proof}
Let $h$ be the minimum of the $y$-coordinates of the points in $A$.
If $h\geq 1$, then $W(A)=\{(x,y): y\leq 1\}$ is a halfplane bounded by the line
$y=1$, so the lemma trivially holds. In the remainder of the proof,
assume that $0<h<1$.

\noindent\textbf{(1)}
Since $0<h<1$, the halfplane below the horizontal line $y=h$ lies in
the interior of $W(A)$ (as every point below this line is in the
interior of a unit disk whose center is below the $x$-axis and whose
interior is disjoint from $A$). Property~1 follows.

\noindent\textbf{(2)}
Let $p\in \partial \text{hull}(A)$. Then $p$ lies on the boundary of a
unit disk $D_p$ whose center is below the $x$-axis (and whose interior
is disjoint from $A$). In particular $D_p\subset W(A)$. The vertical
line segment from $p$ to the $x$-axis lies in $D_p$, hence in $W(A)$.
Consequently, $W(A)$ contains the vertical downward ray emanating from $p$. Property~2 follows.

\noindent\textbf{(3)}
Let $D\in \mathcal{D}$. Suppose, to the contrary, that the
intersection $D\cap (\partial \text{hull}(A))$ has two or more
components. By property~1, the $x$-coordinates of the components are
disjoint intervals, and the components have a natural left-to-right
ordering. Let $p_1$ be the rightmost point in the first component, and
let $p_2$ be the leftmost point in the second component. Clearly
$p_1,p_2\in \partial D$. Let $q$ be an arbitrary point in $\partial
\text{hull}(A)$ between $p_1$ and $p_2$. Then $q$ lies on the boundary
of a unit disk $D_q$ whose center is below the $x$-axis (and whose
interior is disjoint from $A$). Since $D_q\subset W(A)$, neither $p_1$
nor $p_2$ is in the interior of $D_q$. Since the center of $D_q$ is
below the $x$-axis, $\partial D_q$ contains two interior-disjoint
circular arcs between $q$ and the $x$-axis; and both arcs must cross
$\partial D$. We have found two intersection points in $\partial D\cap
\partial D_q$ above the $x$-axis, contradicting
Lemma~\ref{lem:axiom}. This completes the proof of Property~3.

\noindent\textbf{(4)}
Let $D\in \mathcal{D}$ such that $A\cap D\neq \emptyset$. By
continuously translating $D$ vertically down until its interior is
disjoint from $A$, we obtain a unit disk $D'$ such that $A\cap
\text{int}(D')=\emptyset$ but $A\cap \partial D' \neq
\emptyset$. Since the center of $D'$ is vertically below the center of
$D$, we have $A\cap \partial D'\subset A\cap D$ and $D'\subset
W(A)$. This implies that $A\cap \partial D'\subset \partial
\text{hull}(A)$, as required.
\end{proof}

\begin{lemma}\label{lem:epsilon}
Consider the set system $(A,\mathcal{D})$ defined above.
For every $\eps\in (0,\frac23)$, we can construct an $\eps$-net $N=\{v_1,\ldots , v_k\}\subset A$,
labeled by increasing $x$-coordinates, such that
\begin{enumerate}\itemsep0pt
\item $|N|\leq \lfloor 2/\eps\rfloor$;
\item $N\subset \partial \text{hull}(A)$;
\item for every $D\in \mathcal{D}$, the points in $D\cap N$ are consecutive in $N$; and
\item for every $D\in \mathcal{D}$, $|N\cap D|\geq 5$ implies $|A\cap D|\geq 2\eps |A|$.
\end{enumerate}
\end{lemma}
\begin{proof}
Let $M=A\cap \partial \text{hull}(A)$ be the set of points in $A$
lying on the boundary of $\text{hull}(A)$.
By Lemma~\ref{lem:hull}(4), if a unit disk $D\in \mathcal{D}$ contains
any point in $A$, it contains a point from $M$.
Consequently $M$ is an $\eps$-net for $(A,\mathcal{D})$ for every
$\eps>0$.
For a given $\eps>0$, let $N=N_\eps$ be a minimal subset of $M$ that
is an $\eps$-net for $(A,\mathcal{D})$ (obtained, for example, by successively
deleting points from $M$ while we maintain an $\eps$-net).

Let $N=\{v_1,\ldots , v_k\}$, where we label the elements in $N$ by
increasing $x$-coordinates.
For notational convenience, we introduce a point $v_0\in \partial
\text{hull}(A)$ on a vertical line one unit left of $v_1$, and
$v_{k+1}\in \partial \text{hull}(A)$ on a vertical line one unit right
of $v_k$.
For $i=1,\ldots k$, the minimality of $N$ implies that $N\setminus \{v_i\}$ is not an $\eps$-net,
and so there exists a unit disk $D \in \mathcal{D}$ such that
$|A\cap D|\geq \eps |A|$ and $D \cap N=\{v_i\}$. Let $D_i \in \mathcal{D}$ be such a disk,
with $|A\cap D_i|\geq \eps |A|$ and $D_i\cap N=\{v_i\}$.
By Lemma~\ref{lem:hull}(3), $D_i$ contains a connected arc of the $x$-monotone curve
$\partial \text{hull}(A)$, but $D_i$ contains neither $v_{i-1}$ nor $v_{i+1}$.
In particular, the $x$-coordinate of every point in $A\cap D_i$
lies between that of $v_{i-1}$ and $v_{i+1}$. Consequently, every point in $A$
lies in at most two disks $D_i$, $1\leq i\leq k$. It follows that
\[k\cdot \eps |A|
=\sum_{i=1}^k \eps |A|
\leq \sum_{i=1}^k |A\cap D_i|
\leq 2|A|,\]
hence $k\leq \lfloor 2/\eps\rfloor$. This proves (i).

\begin{figure}[htbp]
 \centering
 \includegraphics[width=0.75\textwidth]{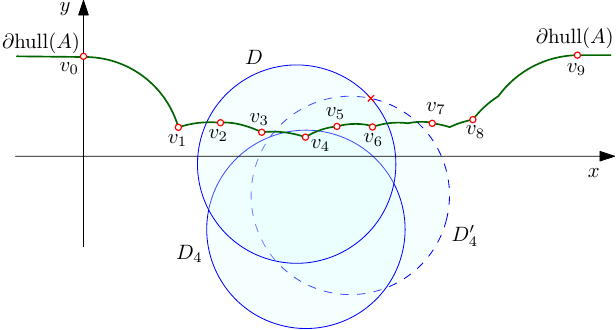}
 \caption{Illustration for the proof of Lemma~\ref{lem:epsilon}(iv) with $i=2$ and $j=4$.
 A unit disk $D$ with $D\cap N=\{v_2,v_3,v_4,v_5,v_6\}$,
 and a unit disk $D_4$ with $v_4\in D_4$ and $v_3,v_5\notin D_4$.
 A hypothetical unit disk $D_4'$ (dashed) such that $v_4\in D_4'$, and
 $\partial D'_4\cap \text{hull}(A)$ crosses $\partial D\cap \text{hull}(A)$.
 }
  \label{fig:schamtic1}
\end{figure}

By construction, we have $N\subset M\subset \partial \text{hull}(A)$,
which confirms (ii), and (iii) follows from Lemma~\ref{lem:hull}(3).
It remains to prove (iv); refer to Fig.~\ref{fig:schamtic1}.
Assume that $D\in \mathcal{D}$ and $|N\cap D|\geq 5$.
By~(iii), we may assume that $D$ contains five consecutive points in $N$, say, $v_i,\ldots , v_{i+4}$.
For $j\in \{i+1,i+2,i+3\}$, consider the disk $D_j\in \mathcal{D}$ defined above,
where $v_j\in D_j$ but $v_{j-1},v_{j+1}\notin D_j$. In particular,
$D_j\cap (\partial \text{hull}(A))$ lies between $v_{j-1}$ and $v_{j+1}$.
By Lemma~\ref{lem:axiom}, the circular arcs $\partial D\cap \text{hull}(A)$
and $\partial D_j\cap \text{hull}(A)$ cross at most once.
However, if they cross once, then $D_j$ contains
one of the endpoints of $D\cap (\partial \text{hull}(A))$,
and by Lemma~\ref{lem:hull}(3) it contains
$\{v_i,\ldots, v_j\}$ or $\{v_j,\ldots , v_{i+4}\}$,
which is a contradiction. We conclude that
$\partial D\cap \text{hull}(A)$ and $\partial D_j\cap \text{hull}(A)$ do not cross.
Consequently, $D_j\cap \text{hull}(A)\subset D\cap \text{hull}(A)$, hence $A\cap D_j\subset A\cap D$.
As noted above, $|A\cap D_j|\geq \eps |A|$. Furthermore,
$A\cap D_{i+1}$ and $A\cap D_{i+3}$ are disjoint as they are on opposite sides of
the vertical line passing through $v_{i+2}$. Thus we obtain
$|A\cap D|\geq |A\cap (D_{i+1}\cup D_{i+3})|
\geq  |A\cap D_{i+1}|+|A\cap D_{i+3}|
\geq 2\eps |A|$, as claimed.
\end{proof}

Let $A$ and $B$ be two disjoint point sets above and below the $x$-axis, respectively.
Denote by $U(A,B)$ the unit disk graph on $A \cup B$ and by $G(A,B)$
the bipartite subgraph of $U(A,B)$ consisting of all edges between $A$ and $B$.

\begin{lemma}\label{lem:bipartite+}
Let $P=A \cup B$ be a set of $n$ points in the plane such that
$\diam(A) \leq 1$, $\diam(B) \leq 1$,
and $A$ (resp., $B$) is above (resp., below) the $x$-axis.
Then there is a subgraph $H$ of $U(A,B)$ with $O(n\log n)$ edges such that
for every edge $ab$ of $G(A,B)$, $H$ contains a path of length at most $2$ between $a$ and $b$.
\end{lemma}
\begin{proof}
Our proof is constructive.
For every point $b\in B$, let $D_b$ be the unit disk centered at $b$.
Consider the set system $(A,\mathcal{B})$, where $\mathcal{B}=\{D_b:b\in B\}$.
We partition the set of disks $\mathcal{B}$ into
$O(\log n)$ subsets based on the number of points of $A$ contained in the disks.
For every $i=1,\ldots , \lceil \log n\rceil$, let
\[\mathcal{B}_i=\left\{D\in \mathcal{B}: \frac{|A|}{2^i}\leq |A\cap D|< \frac{|A|}{2^{i-1}} \right\}.\]

For every $i=1,\ldots , \lceil \log n\rceil$, let $\eps_i=\frac{1}{2^i}$.
Lemma~\ref{lem:epsilon} yields an $\eps_i$-net $N_i\subset A$
of size at most $\lfloor 2/\eps_i\rfloor = 2^{i+1}$ for $(A,\mathcal{B}_i)$.

\begin{figure}[htbp]
 \centering
 \includegraphics[width=0.99\textwidth]{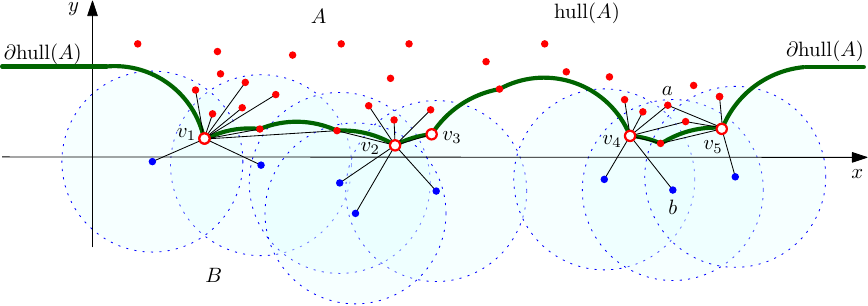}
 \caption{Set $A$ (resp., $B$) is above (resp., below) the $x$-axis.
  The points in an $\eps_i$-net $N_i=\{v_1\ldots ,v_5\}$
  are marked with hollow dots. The graph $H_i$ is a union of stars
  centered at $v_1,\ldots , v_5$.
  (To avoid clutter, the depicted point set does not meet conditions
  $\diam(A) \leq 1$ and $\diam(B) \leq 1$ of Lemma~\ref{lem:bipartite+}.)}
  \label{fig:twohop}
\end{figure}

We construct the graph $H$ as a union of stars; see Fig.~\ref{fig:twohop} for an illustration.
For every $i=1,\ldots , \lceil \log n\rceil$ and every $v\in N_i$,
we create a star centered at $v$ as follows.
Let $B_i(v)$ be the set of points $b\in B$ such that
$D_b\in \mathcal{B}_i$ (that is, $|A|/2^i\leq |A\cap D_b|< |A|/2^{i-1}$),
$v\in D_b$, and $v$ is the leftmost point in $N_i\cap D_b$.
Let $A_i(v)$ be the set of points $a\in A$ contained in
unit disks centered at some point in $B_i(v)$.
Let $S_i(v)$ be the star on $A_i(v)\cup B_i(v)$ centered at $v$.
By construction, every point in $B_i(v)$ is at distance at most 1 from $v$,
and $\diam(A)\leq 1$; this implies that $S_i(v)$ is a subgraph of $U(A,B)$.
Let $H_i$ be the union of all stars centered at vertices in $N_i$;
and let $H$ be the union of the graphs $H_i$ for $i=1,\ldots, \lceil \log n \rceil$.
Note that $H$ is a union of stars in $U(A,B)$, hence a subgraph of $U(A,B)$.

To prove correctness, we show that for every edge $ab$ of $G(A,B)$ (with $a \in A$, $b \in B$),
$H$ contains a path of length $2$ between $a$ and $b$.
Since $ab$ is an edge of $G(A,B)$, we have $|ab|\leq 1$ hence $a\in D_b$.
There exists an index $i\in \{1,\ldots , \lceil \log n\rceil\}$ for which $D_b\in \mathcal{B}_i$.
As $|A\cap D_b|\geq |A|/2^i =\eps_i |A|$, and $N_i$ is an $\eps_i$-net for $(A,\mathcal{B}_i)$,
we have $D_b\cap N_i \neq\emptyset$.
Let $v$ be the leftmost point in $D_b\cap N_i$.
Then by construction $a\in A_i(v)$ and $b\in B_i(v)$.
If $a=v$, then the star $S_i(v)$ contains the edge $ab$,
otherwise $S_i(v)$ contains the path $a,v,b$ of length 2.

It remains to derive an upper bound on the number of edges in $H$.
We claim that $H_i$ has $O(n)$ edges for all $i=1,\ldots , \lceil \log n\rceil$,
which implies that $H$ has $O(n\log n)$ edges overall.

Let $b\in B$. There is a unique index $i$ such that
$|A|/2^i\leq |A\cap D_b|<|A|/2^{i-1}$;
and there is a unique leftmost point $v$ in $N_i\cap D_b$.
Therefore, $b$ is a leaf of only one star $S_i(v)$,
and so its degree is at most 1 in $H_i$, hence in $H$.

Let $i\in \{1,\ldots ,\lceil \log n\rceil\}$.
Assume that $N_i=\{v_1,\ldots, v_k\}$ is sorted by increasing $x$-coordinates.
We also introduce points $v_0$ and $v_{k+1}$ on $\partial \text{hull}(A)$ as
specified previously.

Let $a\in A$; refer to Fig.~\ref{fig:twohop}.
Assume that $a$ is in a star $S_i(v_j)$ for some $v_j\in N_i$.
Assume further that the $x$-coordinate of $a$ is between that of $v_{\ell-1}$ and $v_\ell$
for some $\ell\in \{1,\ldots ,k+1\}$.
Since $a$ is in $S_i(v_j)$, there exists a point $b\in B$ such that
$a\in D_b$, $D_b\in \mathcal{B}_i$, and $v_j$ is the leftmost point in $D_b\cap N_i$.
Since $D_b\in \mathcal{B}_i$, we have $|A\cap D_b|<2\eps_i|A|$.

By Lemma~\ref{lem:epsilon}(iv), $D_b$ contains at most 4 points from the net $N_i$.
In particular, the unit circle $\partial D_b$ intersects $\partial \text{hull}(A)$
in two points: once between $v_{j-1}$ and $v_j$, and once between $v_j$ and $v_{j+4}$.
Consequently, $0\leq \ell-j\leq 4$, thus $a$ is in at most 5 possible stars $S_i(v_j)$,
$v_j\in N_i$. It follows that $H_i$ has at most $5|A|+|B|\leq 5n$ edges, as required.
\end{proof}

We now consider the general case.
\begin{theorem} \label{thm:2hop+}
  Every $n$-vertex unit disk graph has a (possibly nonplane)
  $2$-hop spanner with $O(n\log n)$ edges.
\end{theorem}
\begin{proof}
Let $P$ be a set of $n$ points in the plane.
Consider a tiling of the plane with regular hexagons of unit diameter; and assume that
no point in $P$ lies on the boundary of any hexagon. Let $\T$ be the set of nonempty
hexagons. Then $P$ is partitioned into $O(n)$ sets $\{P \cap \sigma: \sigma \in \T\}$.
As noted in Section~\ref{ssec:5hop}, for every $\sigma \in \T$, there are $18$ other cells
within unit distance; see~Fig.~\ref{fig:hex}\,(left).

For each cell $\sigma \in \T$, choose an arbitrary vertex $v_\sigma \in P \cap \sigma$,
and create a star $S_\sigma$ centered at $v_\sigma$ on the vertex set $P \cap \sigma$.
The overall number of edges in all stars $S_\sigma$, $\sigma \in \T$, is
\[ \sum_{\sigma \in \T}(|P \cap \sigma|-1) = n-|\T| \leq n. \]
For every pair of cells $\sigma_i,\sigma_j \in \T$, where $d(\sigma_i,\sigma_j) \leq 1$,
consider the bipartite graph $G_{i,j} =G(P \cap \sigma_i, P \cap \sigma_j)$.
By Lemma~\ref{lem:bipartite+}, there is a graph $H_{i,j}$
of size
\[ O\big((|P\cap \sigma_i| + |P\cap \sigma_j|)\log (|P\cap \sigma_i| +
|P\cap \sigma_j|)\big) = O\big((|P \cap \sigma_i| +
|P \cap \sigma_j|)\log n\big). \]
Since every vertex appears in at most $18$ such bipartite graphs,
the total number of edges in these graphs is at most
$O \left( \sum_{\sigma \in \T} |P\cap \sigma|\log n \right) =O(n\log n)$.

We show that the union of the stars $S_\sigma$, $\sigma \in \T$, and the graphs $H_{i,j}$
is a $2$-hop spanner. Let $ab$ be an edge of the unit disk graph. If both $a$ and $b$
are in the same cell, say $\sigma \in \T$, then $ab$ is an edge in the star or
the star $S_\sigma$ contains the path $a,v_\sigma,b$.
Otherwise, $a$ and $b$ lie in two distinct cells, say $\sigma_i,\sigma_j \in \T$,
such that $d(\sigma_i,\sigma_j) \leq |ab| \leq 1$. By Lemma~\ref{lem:bipartite+}
(where the role of the $x$-axis is taken by any separating line),
$H_{i,j}$ contains a path of length at most $2$ between $a$ and $b$, as required.
\end{proof}

\section{Lower bounds for plane hop spanners}  \label{sec:lb}

A trivial lower bound of $2$ for the hop stretch factor of plane subgraphs of UDGs
can be easily obtained by taking the four corners of a square of side-length $\frac12$.
In this case, the UDG is the complete graph but a plane spanner cannot contain both diagonals
of the square. Our main result in this section is
a lower bound of $4$ for sufficiently large $n$ (cf.~Theorem~\ref{thm:many-points}).
We begin with a lower bound of $3$
that holds already for $n=8$.

\begin{theorem}
\label{thm:8points}
For every $n\geq 8$, there exists an $n$-element point set $S$ on a circle
such that every plane hop spanner on $S$ has hop stretch factor at least $3$.
\end{theorem}
\begin{proof}
Let $P=\{p_1,\ldots, p_8\}$ be a set of $8$ successive points on a circle of radius $r \geq 1$,
so that $p_1p_8$ is a horizontal chord,
$|p_2p_3|=|p_3p_4|=|p_4p_5|=|p_5p_6|=|p_6p_7|$, $|p_1p_2|=|p_7p_8|=1.1 |p_2p_3|$,
$|p_1p_4|<1$, and  $|p_2p_6|=|p_3p_7|=1$.
The UDG of $P$ is shown in Fig.~\ref{fig:lb3}\,(left). Note that  $|p_1p_5| =|p_4p_8|>1$;
and that the orthogonal bisector of $p_1p_8$ is a vertical axis of symmetry.
Since $P$ is in convex position we may assume that $p_ip_{i+1} \in E'$ for $i=1,\ldots,7$.
Suppose that $G'=(P,E')$ is a plane hop spanner with hop stretch factor~$2$.
Define the \emph{span} of an edge $p_i p_j$ ($i<j$), as $j-i$.
We distinguish between two cases depending on whether $E'$ contains at least one edge of span $2$
whose endpoints are in $\{p_2,\ldots,p_7\}$.
    	
\begin{figure}[htbp]
\centering
\includegraphics{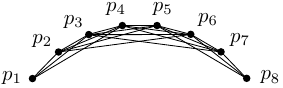}
\hspace{1cm}
\includegraphics{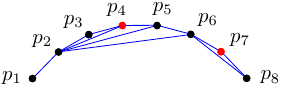}
\caption{Left: the $8$-element point set $P$ and its UDG. Right: a $3$-hop plane spanner
of $P$; for the hop distance between the two red points, $p_4$ and $p_7$, is 3.}
\label{fig:lb3}
\end{figure}
	
\emph{Case 1: $E'$ contains at least one edge of span $2$ whose endpoints are in $\{p_2,\ldots,p_7\}$.}
Assume first that $p_3p_5 \in E'$ or $p_4p_6 \in E'$. Assume w.l.o.g.\ that $p_3p_5 \in E'$.
Since $h(p_1,p_4) \leq 2$, we have $p_1p_3 \in E'$.
Since $h(p_2,p_6) \leq 2$, we have $p_3p_6 \in E'$.
Since $h(p_4,p_7) \leq 2$, we have $p_3p_7 \in E'$.
Then $\rho(p_5,p_8)$ has at least $3$ hops, a contradiction.

We can subsequently assume that $p_3p_5, p_4p_6 \notin E'$.
Assume next that $p_2p_4 \in E'$ or $p_5p_7 \in E'$. Assume w.l.o.g.\ that $p_2p_4 \in E'$.
Since $h(p_3,p_6) \leq 2$, we have $p_2p_6 \in E'$.
Then $\rho(p_4,p_7)$ has at least $3$ hops, a contradiction.

\medskip
\emph{Case 2: $E'$ contains no edge of span $2$ whose endpoints are in $\{p_2,\ldots,p_7\}$.}
Since $h(p_3,p_6) \leq 2$, we have $p_3p_6 \in E'$, $p_2p_6 \in E'$, or $p_3p_7 \in E'$.
If $p_3p_6 \in E'$, then $\rho(p_2,p_5)$ has at least $3$ hops, a contradiction.
Assume w.l.o.g.\ that $p_2p_6 \in E'$.
Then $\rho(p_1,p_4)$ has at least $3$ hops, a contradiction.

\medskip
Thus, we have shown that every plane hop spanner on $P$ has hop stretch factor of at least $3$.
For every $n \geq 8$, we can add $n-8$ points on the circle beyond $p_8$ such that every plane hop spanner
on the resulting set $S$ of $n$ points has hop stretch factor of at least~$3$.
\end{proof}

We next derive a better bound assuming that $n$ is sufficiently large.

\begin{theorem} \label{thm:many-points}
For every sufficiently large $n$, there exists an $n$-element point set $P$ on a circle
such that every plane hop spanner on $P$ has hop stretch factor at least $4$.
\end{theorem}
\begin{proof}
  Consider a set $P$ of $n$ points that form the vertices of regular $n$-gon
  $R$ inscribed in a circle $C$, where the circle is just a bit larger than the
  circumscribed circle of an equilateral triangle of unit edge length.
  Formally, for a given $\eps \in (0,1/50)$, set $n =\lceil 2\eps^{-1} \rceil$ and choose
  the radius of $C$ such that every sequence of $\left(\frac13 -\eps\right) n $
  consecutive points from $P$ makes  a subset of diameter at most $1$;
  and any larger sequence  makes a subset of diameter larger than $1$.
  Note that $\eps n \geq 2$. (We may set $\eps =0.02$, which yields $n=100$.)

  The short circular arc between two consecutive vertices of $R$
  is referred to as an \emph{elementary arc}. (Its center angle is $2\pi/n$.)
  If $A$ is a set of elementary arcs, $X(A)$ denotes its set of endpoints;
  obviously $|X(A)| \geq |A|$, with equality when $A$ covers the entire circle $C$.

  Suppose, for the sake of contradiction, that the unit disk graph $G$ has a plane
  subgraph $G'$ with hop number at most $3$.
  First, augment $G'$ to a maximal noncrossing subgraph of $G$, by successively adding edges
  from $G \setminus G'$ that do not introduce crossings.
  Adding edges does not increase the hop number of $G'$, which remains at most $3$.

We define \emph{maximal} edges in $G'$ as follows.
Associate every edge of $G'$ with the shorter circular arc between its endpoints.
Observe that containment between arcs is a partial order (poset).
An edge of $G'$ is \emph{maximal} if the associated arc is maximal in this poset.
Due to planarity, if two arcs overlap, then one of the arcs contains the other.
Hence the maximal edges correspond to nonoverlapping arcs.
As such, the maximal edges form a convex cycle, \ie, a convex polygon
$Q=p_1,p_2,\ldots,p_k$. Refer to Fig.~\ref{fig:f1}.
By the choice of $C$, we have $k \geq 4$.
Each edge of the polygon $Q$ determines a set of points, called \emph{block},
that lie on the associated circular arc (both endpoints of the edge are included).
Since the length of each edge of $Q$ is at most $1$, the restriction of $G'$
to the vertices in a block is a triangulation.
\begin{figure}[htbp]
\centering
\includegraphics[scale=0.63]{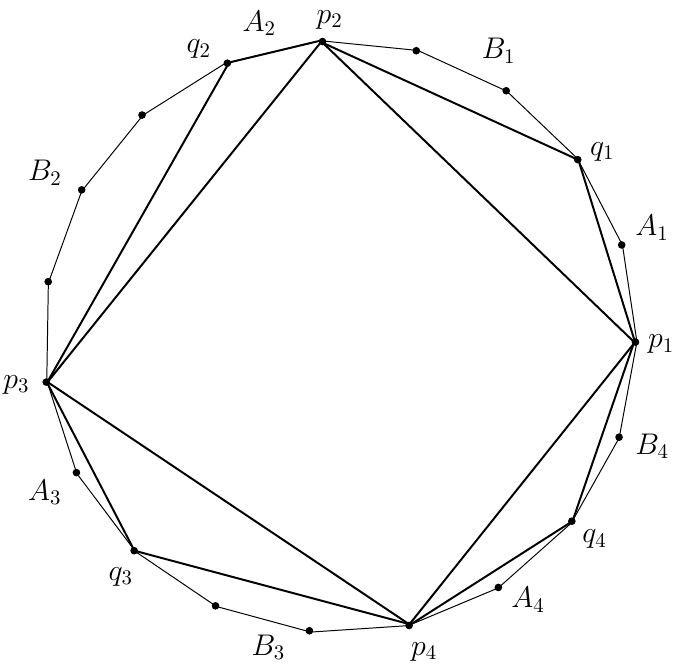}
\caption{The partition induced by the blocks for $n=19$ and $k=4$.
  The edges $p_i p_{i+1}$ are maximal edges of $G'$ and $\Delta{p_i p_{i+1} q_i}$ is the
  unique triangle adjacent to $p_i p_{i+1}$ in the triangulation of the $i$th block.
Since $n=19$ is small, the figure only illustrates the notation used in the proof of
Theorem~\ref{thm:many-points}; $|A_1|=2$, $|B_1|=3$,  $|A_2|=1$, $|B_2|=4$, etc.}
\label{fig:f1}
\end{figure}

Let $A_i \cup B_i$ be the sets of elementary arcs in counterclockwise order
covering the $i$th block such that $A_i$ and $B_i$ are separated by a common vertex $q_i$,
where the triangle $\Delta{p_i p_{i+1} q_i}$ is the (unique) triangle
adjacent to the chord $p_i p_{i+1}$ in the triangulation of the $i$th block
(where addition is modulo $k$, so that $k+1=1$).
In particular, $q_i$ is the last endpoint of an elementary arc in $A_i$ and the first
endpoint of an elementary arc in $B_i$, in counterclockwise order.
As such, we have
\begin{equation} \label{eq:sum}
   \sum_{i=1}^k (|A_i| + |B_i|) =n.
\end{equation}

By definition, we have
\begin{equation} \label{eq:each}
|A_i| + |B_i| \leq \left(\frac13 -\eps\right) n, \ \ \text{ for } i=1,\ldots,k.
\end{equation}

By the maximality of the blocks in $G'$, we have
\begin{equation} \label{eq:pair}
(|A_i| + |B_i|) + (|A_{i+1}| + |B_{i+1}|) \geq
\left(\frac13 -\eps\right) n, \ \ \text{ for } i=1,\ldots,k.
\end{equation}

By the maximality of $G'$, we also have $k \leq 6$, since otherwise
an averaging argument would yield two adjacent blocks, say, $i$ and $i+1$,
that can be merged by adding one chord of length at most $1$ and so
that the merged sequence of points has size at most
\[ |A_i| + |B_i| + |A_{i+1}| + |B_{i+1}|  \leq \frac{2n}{7} < \left(\frac13 -\eps\right)n, \]
which would be a contradiction. We next prove the following inequality:
\begin{equation} \label{eq:neighborhood}
  |B_i| + |A_{i+1}| > \left(\frac13 -3\eps\right) n, \ \ \text{ for } i=1,\ldots,k.
\end{equation}

Suppose for contradiction that $|B_i| + |A_{i+1}| \leq \left(\frac13 -3\eps\right) n$
holds for some $i$. Consider the $\eps n$ elementary arcs preceding the arcs in $B_i$
and the $\eps n$ elementary arcs following the arcs in $A_{i+1}$, in counterclockwise order.
Denote these sets of arcs by $U_i$ and $V_i$, respectively ($|U_i|=|V_i|=\eps n$).
Recall that $\eps n \geq 2$ and thus $|X(U_i)|, |X(V_i)| \geq |U_i| = \eps n \geq 2$.

We claim that there exist $u \in X(U_i)$ and  $v \in X(V_i)$ such that $|uv| \leq 1$ and
$h(u,v) \geq 4$. Indeed, $\diam(X(U_i \cup B_i \cup A_{i+1} \cup V_i)) \leq 1$
since $X(U_i \cup B_i \cup A_{i+1} \cup V_i)$ contains at most
\[  \left(\frac13 -3\eps\right) n + 2\eps n \leq \left(\frac13 -\eps\right) n \]
consecutive points.
This proves the first part of the claim for any $u \in X(U_i)$ and  $v \in X(V_i)$.
For the second part,
we can take $u$ as one of the two vertices preceding $q_i$ that is not $p_i$,
and similarly we can take $v$ as one of the two vertices following $q_{i+1}$ that is not $p_{i+2}$.
With this choice, we have $h(u,p_{i+1}) \geq 2$ and $h(p_{i+1},v) \geq 2$, and
$\rho(u,v)$ passes through $p_{i+1}$. Consequently,
\[h(u,v) \geq h(u,p_{i+1}) + h(p_{i+1},v) \geq 2+2=4. \]
We have reached a contradiction, which proves~\eqref{eq:neighborhood}.
The summation of~\eqref{eq:neighborhood} over all $i=1,\ldots ,k$,
in combination with~\eqref{eq:sum} and the inequality $k \geq 4$ yields
\[
n = \sum_{i=1}^k \, (|A_i| + |B_i|)
=\sum_{i=1}^k  \, (|B_i| + |A_{i+1}|)
\geq k \left(\frac13 -3\eps\right) n
\geq 0.27 \, k n \geq 1.08 \,n.
\]
This last contradiction completes the proof of the theorem.
\end{proof}

\paragraph{An upper bound for points on a circle.}
For many problems dealing with finite point configurations in the plane, points in convex position
or on a circle may allow for tighter bounds; see, \eg,
\cite{du1985steiner,du1987steiner,mulzer2004minimum,sattari2019improved}.
We show that the lower bound of $4$ for points on a circle is tight in this case.

\begin{theorem}
\label{thm:circle}
  For every finite point set $S$ on a circle $C$, there exists a plane $4$-hop spanner.
\end{theorem}

\begin{proof}
Let $C$ be a circle with center $o \in \RR^2$ and radius $r>0$. Let $S$
be a set of $n$ points on $C$, and let $G=G(S)$ be the corresponding UDG.
We may assume w.l.o.g.\ that $G$ is connected.
If $r \leq 1/2$, then $G=K_n$, we set $G'=K_{1,n-1}$, \ie, a star centered at
an arbitrary point. This yields $h(s,s') \leq 2$ for every $s,s' \in S$.
We therefore subsequently assume that $r> 1/2$; this implies that no edge of
$G$ passes through $o$.
	
Let $\gamma \subset C$ be a shortest arc of $C$ covering the points in
$S=\{s_1,s_2,\ldots,s_n\}$, where the points are labeled counterclockwise on $\gamma$.
We claim that $|s_i s_{i+1}| \leq 1$, for $i=1,\ldots,n-1$.
Indeed, let $1 \leq i \leq n-1$ be the smallest index such that $|s_i s_{i+1}| >1$.
Then $|s_1 s_n| \geq |s_i s_{i+1}| >1$ and therefore
$\{s_1,\ldots,s_i\}$ and $\{s_{i+1},\ldots,s_n\}$ are disconnected in $G$,
a contradiction. We construct a plane subgraph $G'=(S,E')$ of $G$ in two phases,
and then show the $G'$ is a 4-hop spanner for $S$.

In the first phase, we incrementally construct a polygonal chain $Q=p_1,p_2,\ldots,p_k$,
on a subset of $k$ elements of $S$ with the vertices chosen counterclockwise
by a greedy algorithm starting with $p_1=s_1$ ($k$ is determined by the algorithm).
The polygon $Q$ will be part of the plane graph $G'$; the following properties will be satisfied.
\begin{itemize} \itemsep 1pt
\item $p_i \in S$, for $i=1,\ldots,k$,
\item $|p_i p_{i+1}| \leq 1$, $i=1,\ldots,k-1$.
\end{itemize}
In the current step, assume that $p_i$ has already been selected;
here $p_i$ precedes $s_n$. The algorithm checks subsequent points
counterclockwise on $C$, say $s_j,s_{j+1},\ldots$
As noted above, since $G$ is connected, we have $|p_i s_j| \leq 1$.
The algorithm selects $p_{i+1}=s_{j+h}$, where $h \geq 0$ is the largest index
such that $|p_i s_{j+h'}| \leq 1$ for $h'=0,1,\ldots,h$, \ie, for \emph{all}
successive points until $s_{j+h}$; or $p_{i+1}=s_n$, if the last point is reached.
If $p_{i+1}$ precedes $s_n$, the algorithm updates $i \gets i+1$ and
continues with the next iteration; if $p_{i+1}=s_n$, we set  $k:=i$.
When this process terminates, $k$ is set.

\begin{figure}[htbp]
\centering
	\includegraphics[scale=0.5]{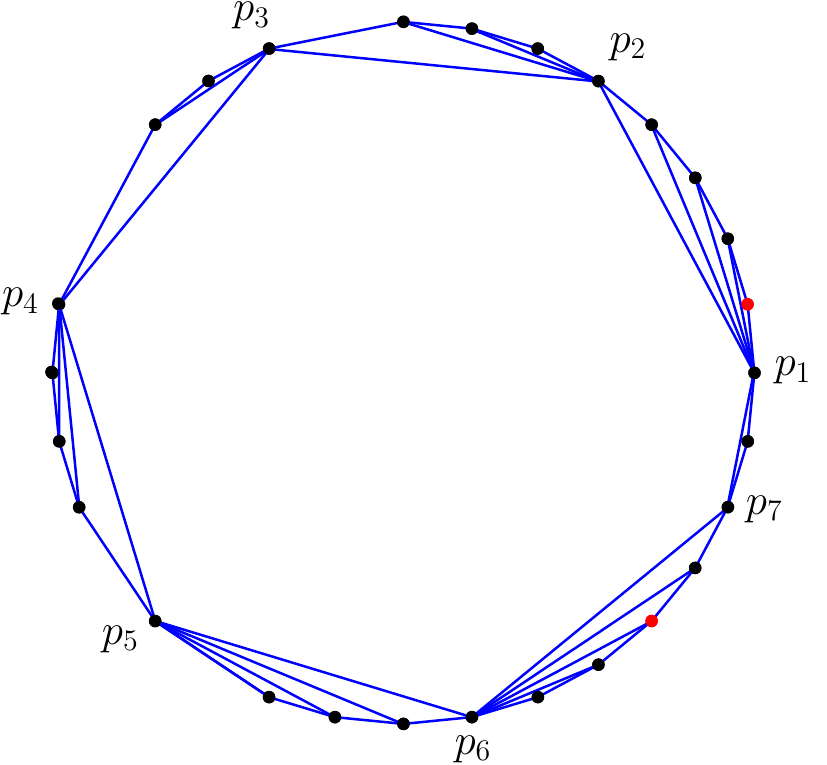}
	\caption{An example of the $4$-hop spanner constructed by the greedy algorithm;
	$P=p_1,\ldots,p_7$ is a closed chain.}
	\label{fig:test}
\end{figure}

If $|p_k p_1| \leq 1$, the edge $p_k p_1$ is added to close the chain,
\ie, $Q$ is a convex polygon whose $k$ edges belong to $E$, in particular, $p_k p_1 \in E$;
note that there may be points of $S$ on the arc $\arc{p_k p_1}$.
It is possible that $|p_k p_1|>1$, in which case $Q=p_1,\ldots,p_k$ is an open chain
with $k-1$ edges. In this case there are no other points of $S$ on the arc $\arc{p_k p_1}$.
Each edge of the chain $Q$ determines a set of points called \emph{block}
(endpoints of the edge are included). Depending on whether the chain $Q$
is open or closed, there are either $k-1$ blocks or $k$ blocks.
		
In the second phase, for every edge $p_i p_{i+1} \in E$ (with wrap around),
we connect $p_i$ with all other points (if any) in that block
(\ie, create a \emph{star} whose apex is $p_i$); refer to Fig.~\ref{fig:test}
for an example. This completes the construction of the plane graph $G'=(S,E')$.
	
It remains to analyze its hop factor of $G'$. Let $uv \in E$ be any edge of $G$;
we may assume w.l.o.g.\ that $uv$ is horizontal and lies below the center $o$.
Refer to Fig.~\ref{fig:f8}\,(right).
We show that $uv$ can have at most one edge of $Q$ strictly below it. 
Suppose that $e= p_i p_{i+1} \in E'$ is an edge of the polygon $Q$ that lies
strictly below $uv$. We claim that $i=k$, \ie, $e= p_k p_1$ and so this edge is unique
if this occurs. Note that if $e= p_k p_1 \in E$, then the chain $Q$ is closed.

\begin{figure}[htbp]
	\centering
	\includegraphics[scale=0.7]{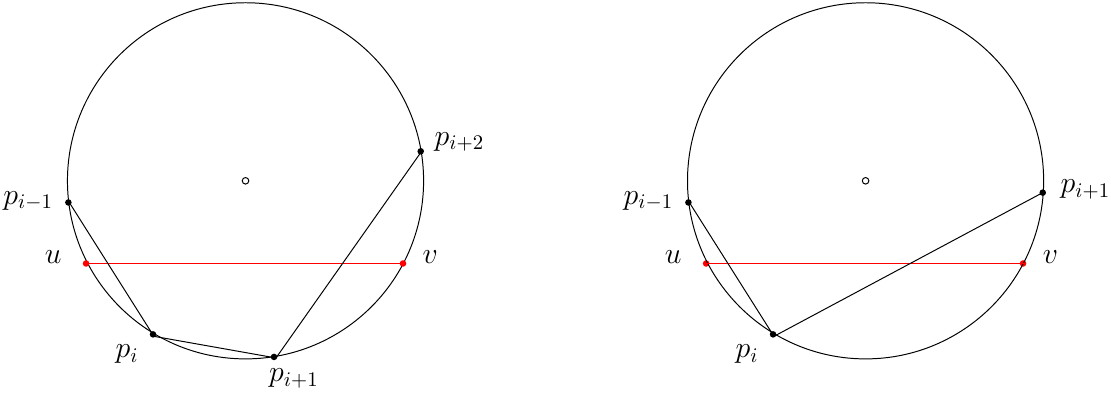}
	\caption{Left: the path connecting $u$ and $v$ is $u p_{i-1} p_i p_{i+1} v$.
		Right: the path connecting $u$ and $v$ is $u p_{i-1} p_i v$.}
	\label{fig:f8}
\end{figure}

Assume that $i \neq k$. Since $uv$ is below the horizontal diameter of $C$,
we have $|p_i p_{i+1}| < |p_i v| <|uv| \leq 1$,
and thus the greedy algorithm would have chosen $v$ or another vertex beyond $v$
counterclockwise, instead of $p_{i+1}$ as the other endpoint of the edge incident to $p_i$,
a contradiction. This proves the claim.
	
By the claim, the endpoints of every edge $uv \in E$ lie either in the same block,
in two adjacent blocks, or in two blocks that are separated by exactly one other block.
Consequently, $uv$ can be connected by a $h$-hop path, for some $h \leq 4$.
Fig.~\ref{fig:f8}\,(left) shows the case when the endpoints $u,v$ belong to two blocks
that are separated by exactly one other block: the connecting path is $u p_{i-1} p_i p_{i+1} v$.
Fig.~\ref{fig:f8}\,(right) shows the case when the endpoints $u,v$ belong to two adjacent blocks:
the connecting path is $u p_{i-1} p_i v$.
When both $u$ and $v$ belong to the same block of the chain, they are connected either directly
or by a path of length $2$ via the center of the corresponding star.
\end{proof}

\section{The maximum degree of hop spanners cannot be bounded} \label{sec:constant}

It is not difficult to see that dense (abstract) graphs do not admit
bounded degree hop spanners (irrespective of planarity). We start with an observation
regarding the complete UDG $K_n$ and then extend it and show that the maximum degree
of hop spanners of sparse UDGs is also unbounded.

We use the fact that graphs of small diameter \emph{and} maximum degree must be small.
Indeed, a connected graph with diameter at most $D$ and maximum degree is at most~$\Delta \geq 3$
has fewer than $\frac{\Delta}{\Delta-2} \cdot (\Delta-1)^D$ vertices~\cite[Proposition~1.3.3]{Diestel};
and a connected graph with diameter at most $D$ and maximum degree is at most~$2$ has fewer than $2D+2$
vertices. As such,  a connected graph with diameter at most $D$ and maximum degree 
at most~$\Delta \geq 2$ has fewer than $2 \Delta^D$ vertices.

\begin{theorem} \label{thm:K_n}
  For every pair of integers $k \geq 2$ and $\Delta \geq 2$, there exists a set $S$ of
  $n \leq 2\Delta^k$ points such that the unit disk graph
  $G=(S,E)$ on $S$ has no $k$-hop spanner whose maximum degree is at most~$\Delta$.
\end{theorem}
\begin{proof}
Let $S$ be a set of $n$ points in a unit disk. Then the UDG $G$ of $S$ is the complete graph $K_n$.
Suppose, to the contrary, that $G'=(S, E')$ is a $k$-hop
spanner for $G$ with maximum degree at most $\Delta$.
Then $h(p,q)\leq k$ for all $p,q\in S$, hence the diameter of $G'$ is at most $k$.
By the above observation we have $n < 2 \Delta^k$, thus we obtain a contradiction if we set
$n = 2 \Delta^k$.
\end{proof}

\begin{theorem} \label{thm:dense}
Let $t \colon \NN \to \NN$, $t(n) \leq n$, be an integer function that tends  to $\infty$ with $n$.
For every pair of integers $k \geq 2$ and $\Delta \geq 2$, there exists $n_0 \in \NN$ such that
for every $n \geq n_0$, there is a set $S$ of $n$ points in the plane such that
\begin{enumerate} \itemsep 0pt
\item [{\rm (i)}] the unit disk graph $G=(S,E)$ on $S$
  has $\Theta(n \cdot t(n))$ edges, and
\item [{\rm (ii)}] $G$ has no $k$-spanner whose maximum degree is at most $\Delta$.
\end{enumerate}
\end{theorem}
\begin{proof}
For a given $t$, partition $n$ points into $\left \lfloor \frac{n}{t} \right \rfloor$
groups of size $t$ and a remaining group (if any) of size
$n - \left \lfloor \frac{n}{t} \right \rfloor t$.
Place the groups in disjoint disks of unit diameter in the plane,
so that the UDG of each group is a complete graph;
and arrange the disks along a line such that
the UDG $G$ has exactly one edge between any two consecutive groups.
Each group of size $t$ induces $\binom{t}{2}=\Theta(t^2)$ edges,
hence $G$ has $\Theta(\frac{n}{t}\cdot t^2+t)=\Theta(nt)$ edges.

Suppose that $G$ has a $k$-hop spanner $G'$ with maximum degree at most $\Delta$.
Then $h(p,q)\leq k$ for all $p,q\in S$ within the same group, hence each group
induces a subgraph of $G'$ of diameter at most $k$.
By the above observation we have $t < 2 \Delta^k$, thus we obtain a contradiction if
we choose $n_0$ such that $t(n) \geq  2 \Delta^k$ for all $n\geq n_0$.
\end{proof}

\section{Conclusions}  \label{sec:conclusion}

We have shown that the UDG of every set of $n$ points in the plane
admits a 5-hop spanner with at most $5.5n$ edges, a 3-hop spanner with
at most $11n$ edges, and a 2-hop spanner with $O(n\log n)$ edges.
The third bound leaves an interesting question:
Are there $n$-element point sets for which every 2-hop spanner has
$\omega(n)$ edges? Recent results show that unit disks may exhibit
surprising behavior~\cite{mcdiarmid2014number,pach2016unsplittable}.

Finding nontrivial lower bounds for the size of $k$-hop spanners
remains an open problem. We mention a few straightforward lower
bounds. Observe that if the girth of an UDG $G$ is $k\geq 4$, then the
only $(k-2)$-hop spanner of $G$ is $G$ itself. In particular, for $n$
points in a section of the square lattice $\ZZ^2$, the UDG has
$(2-o(1))n$ edges, its girth is 4, and so the only 2-hop spanner of $G$ is
$G$ itself. For $n$ points in a section of a hexagonal lattice, the UDG
has $(\frac32-o(1))n$ edges, its girth is 6, and so the only 3- or 4-hop
spanner of $G$ is $G$ itself. Finally, for $n$ points in $\ZZ^2\setminus 2\ZZ^2$,
the UDG has $(\frac43-o(1))n$ edges, its girth is 8, and so the only 5- or
6-hop spanner of $G$ is $G$ itself.

Biniaz~\cite{biniaz2020plane} showed that the UDG of every point set
admits a hop spanner with hop stretch factor at most 341.
For points on a circle, we have improved the upper bound to 4, and
showed that this bound is the best possible. This is the first
nontrivial lower bound for the hop stretch factor of any plane hop
spanner (Theorem~\ref{thm:circle}). Are there point sets for which
every plane hop-spanner has hop stretch factor at least $5$?

In this paper, we considered the UDG of a point set in terms of
Euclidean distance (\ie, $L_2$-norm) in the plane.
We can define UDG over any other norm over $\RR^2$, where the unit
disks are translates of a centrally symmetric convex body.
Estimating the size of hop spanners over arbitrary normed spaces in $\RR^2$
is another problem for consideration.

\bibliographystyle{abbrv}
\bibliography{a}

\begin{thebibliography}{10}

\bibitem{althofer1993sparse}
I.~Alth{\"{o}}fer, G.~Das, D.~P. Dobkin, D.~Joseph, and J.~Soares.
\newblock On sparse spanners of weighted graphs.
\newblock {\em Discret. Comput. Geom.}, 9:81--100, 1993.

\bibitem{alzoubi2003geometric}
K.~M. Alzoubi, X.~Li, Y.~Wang, P.~Wan, and O.~Frieder.
\newblock Geometric spanners for wireless ad hoc networks.
\newblock {\em {IEEE} Trans. Parallel Distrib. Syst.}, 14(4):408--421, 2003.

\bibitem{aronov2008sparse}
B.~Aronov, M.~de~Berg, O.~Cheong, J.~Gudmundsson, H.~J. Haverkort, M.~H.~M.
  Smid, and A.~Vigneron.
\newblock Sparse geometric graphs with small dilation.
\newblock {\em Comput. Geom.}, 40(3):207--219, 2008.

\bibitem{baswana2007simple}
S.~Baswana and S.~Sen.
\newblock A simple and linear time randomized algorithm for computing sparse
  spanners in weighted graphs.
\newblock {\em Random Struct. Algorithms}, 30(4):532--563, 2007.

\bibitem{Baswana2016}
S.~Baswana and S.~Sen.
\newblock {\em Simple Algorithms for Spanners in Weighted Graphs}, pages
  1981--1986.
\newblock Springer, New York, 2016.

\bibitem{biniaz2020plane}
A.~Biniaz.
\newblock Plane hop spanners for unit disk graphs: {S}impler and better.
\newblock {\em Comput. Geom.}, 89:101622, 2020.

\bibitem{bose2005constructing}
P.~Bose, J.~Gudmundsson, and M.~H.~M. Smid.
\newblock Constructing plane spanners of bounded degree and low weight.
\newblock {\em Algorithmica}, 42(3-4):249--264, 2005.

\bibitem{bose2013plane}
P.~Bose and M.~H.~M. Smid.
\newblock On plane geometric spanners: {A} survey and open problems.
\newblock {\em Comput. Geom.}, 46(7):818--830, 2013.

\bibitem{breu1998unit}
H.~Breu and D.~G. Kirkpatrick.
\newblock Unit disk graph recognition is {NP}-hard.
\newblock {\em Comput. Geom.}, 9(1-2):3--24, 1998.

\bibitem{calinescu2006bounded}
G.~C{\u{a}}linescu, S.~Kapoor, and M.~Sarwat.
\newblock Bounded-hops power assignment in ad hoc wireless networks.
\newblock {\em Discret. Appl. Math.}, 154(9):1358--1371, 2006.

\bibitem{catusse2010planar}
N.~Catusse, V.~Chepoi, and Y.~Vax{\`{e}}s.
\newblock Planar hop spanners for unit disk graphs.
\newblock In {\em Proc. 6th Workshop on Algorithms for Sensor Systems, Wireless
  Ad Hoc Networks, and Autonomous Mobile Entities ({ALGOSENSORS})}, volume 6451
  of {\em LNCS}, pages 16--30. Springer, 2010.

\bibitem{das1989triangulations}
G.~Das and D.~Joseph.
\newblock Which triangulations approximate the complete graph?
\newblock In {\em Proc. International Symposium on Optimal Algorithms}, volume
  401 of {\em LNCS}, pages 168--192. Springer, 1989.

\bibitem{Diestel}
R.~Diestel.
\newblock {\em Graph Theory}, volume 173 of {\em Graduate Texts in
  Mathematics}.
\newblock Springer, 5th edition, 2017.

\bibitem{dobkin1990delaunay}
D.~P. Dobkin, S.~J. Friedman, and K.~J. Supowit.
\newblock Delaunay graphs are almost as good as complete graphs.
\newblock {\em Discret. Comput. Geom.}, 5:399--407, 1990.

\bibitem{du1985steiner}
D.~Du, F.~K. Hwang, and S.~C. Chao.
\newblock Steiner minimal tree for points on a circle.
\newblock {\em Proceedings of the American Mathematical Society},
  95(4):613--618, 1985.

\bibitem{du1987steiner}
D.~Du, F.~K. Hwang, and J.~F. Weng.
\newblock Steiner minimal trees for regular polygons.
\newblock {\em Discret. Comput. Geom.}, 2:65--84, 1987.

\bibitem{dumitrescu2016lattice}
A.~Dumitrescu and A.~Ghosh.
\newblock Lattice spanners of low degree.
\newblock {\em Discret. Math. Algorithms Appl.}, 8(3):1650051:1--1650051:19,
  2016.

\bibitem{dumitrescu2016lower}
A.~Dumitrescu and A.~Ghosh.
\newblock Lower bounds on the dilation of plane spanners.
\newblock {\em Int. J. Comput. Geom. Appl.}, 26(2):89--110, 2016.

\bibitem{dutta2019multi}
A.~Dutta, A.~Ghosh, and O.~P. Kreidl.
\newblock Multi-robot informative path planning with continuous connectivity
  constraints.
\newblock In {\em Proc. IEEE International Conference on Robotics and
  Automation ({ICRA})}, pages 3245--3251, 2019.

\bibitem{edelsbrunner1983shape}
H.~Edelsbrunner, D.~G. Kirkpatrick, and R.~Seidel.
\newblock On the shape of a set of points in the plane.
\newblock {\em {IEEE} Trans. Inf. Theory}, 29(4):551--558, 1983.

\bibitem{DBLP:books/el/00/Eppstein00}
D.~Eppstein.
\newblock Spanning trees and spanners.
\newblock In J.~Sack and J.~Urrutia, editors, {\em Handbook of Computational
  Geometry}, chapter~9, pages 425--461. North Holland, 2000.

\bibitem{Erdos64extremalproblems}
P.~Erd\H{o}s.
\newblock Extremal problems in graph theory.
\newblock In {\em Theory of Graphs and its Applications (Proc. Sympos.
  Smolenice, 1963)}, pages 29--36, Prague, 1964. Publishing House of the
  Czechoslovak Academy of Sciences.

\bibitem{gao2005geometric}
J.~Gao, L.~J. Guibas, J.~Hershberger, L.~Zhang, and A.~Zhu.
\newblock Geometric spanners for routing in mobile networks.
\newblock {\em {IEEE} J. Sel. Areas Commun.}, 23(1):174--185, 2005.

\bibitem{kanj2008geometric}
I.~A. Kanj and L.~Perkovic.
\newblock On geometric spanners of {E}uclidean and unit disk graphs.
\newblock In {\em Proc. 25th Symposium on Theoretical Aspects of Computer
  Science ({STACS})}, volume~1 of {\em LIPIcs}, pages 409--420. Schloss
  Dagstuhl, 2008.

\bibitem{kortsarz1994generating}
G.~Kortsarz and D.~Peleg.
\newblock Generating sparse 2-spanners.
\newblock {\em J. Algorithms}, 17(2):222--236, 1994.

\bibitem{kortsarz1998generating}
G.~Kortsarz and D.~Peleg.
\newblock Generating low-degree 2-spanners.
\newblock {\em {SIAM} J. Comput.}, 27(5):1438--1456, 1998.

\bibitem{levcopoulos1992there}
C.~Levcopoulos and A.~Lingas.
\newblock There are planar graphs almost as good as the complete graphs and
  almost as cheap as minimum spanning trees.
\newblock {\em Algorithmica}, 8(3):251--256, 1992.

\bibitem{li2003algorithmic}
X.~Li.
\newblock Algorithmic, geometric and graphs issues in wireless networks.
\newblock {\em Wirel. Commun. Mob. Comput.}, 3(2):119--140, 2003.

\bibitem{DBLP:journals/ijcga/LiW04}
X.~Li and Y.~Wang.
\newblock Efficient construction of low weighted bounded degree planar spanner.
\newblock {\em Int. J. Comput. Geom. Appl.}, 14(1-2):69--84, 2004.

\bibitem{mcdiarmid2014number}
C.~McDiarmid and T.~M{\"{u}}ller.
\newblock The number of disk graphs.
\newblock {\em Eur. J. Comb.}, 35:413--431, 2014.

\bibitem{mitchell2000geometric}
J.~S.~B. Mitchell.
\newblock Geometric shortest paths and network optimization.
\newblock In J.~Sack and J.~Urrutia, editors, {\em Handbook of Computational
  Geometry}, pages 633--701. North Holland, 2000.

\bibitem{mulzer2004minimum}
W.~Mulzer.
\newblock {\em Minimum dilation triangulations for the regular n-gon}.
\newblock Master’s thesis, Freie Universit{\"a}t Berlin, 2004.

\bibitem{mustafa2017epsilon}
N.~H. Mustafa and K.~R. Varadarajan.
\newblock Epsilon-approximations and epsilon-nets.
\newblock In J.~E. Goodman, J.~O'Rourke, and C.~D. T\'oth, editors, {\em
  Handbook of Discrete and Computational Geometry}, chapter~47. CRC Press, Boca
  Raton, FL, 3 edition, 2017.

\bibitem{narasimhan2007geometric}
G.~Narasimhan and M.~H.~M. Smid.
\newblock {\em Geometric Spanner Networks}.
\newblock Cambridge University Press, New York, 2007.

\bibitem{nieberg2008approximation}
T.~Nieberg, J.~L. Hurink, and W.~Kern.
\newblock Approximation schemes for wireless networks.
\newblock {\em {ACM} Trans. Algorithms}, 4(4):49:1--49:17, 2008.

\bibitem{pach2016unsplittable}
J.~Pach and D.~P{\'a}lv{\"o}lgyi.
\newblock Unsplittable coverings in the plane.
\newblock {\em Adv. Math.}, 302:433--457, 2016.

\bibitem{pach2013tight}
J.~Pach and G.~Tardos.
\newblock Tight lower bounds for the size of epsilon-nets.
\newblock {\em J. Amer. Math. Soc.}, 26(3):645--658, 2013.

\bibitem{pach1990some}
J.~Pach and G.~J. Woeginger.
\newblock Some new bounds for epsilon-nets.
\newblock In {\em Proc. 6th ACM Symposium on Computational Geometry ({SoCG})},
  pages 10--15, 1990.

\bibitem{peleg1989graph}
D.~Peleg and A.~A. Sch{\"{a}}ffer.
\newblock Graph spanners.
\newblock {\em Journal of Graph Theory}, 13(1):99--116, 1989.

\bibitem{rajaraman2002topology}
R.~Rajaraman.
\newblock Topology control and routing in ad hoc networks: {A} survey.
\newblock {\em {SIGACT} News}, 33(2):60--73, 2002.

\bibitem{sattari2019improved}
S.~Sattari and M.~Izadi.
\newblock An improved upper bound on dilation of regular polygons.
\newblock {\em Comput. Geom.}, 80:53--68, 2019.

\end{thebibliography}

\end{document}